\documentclass[10pt,letterpaper]{article}
\usepackage[utf8]{inputenc}
\usepackage[margin=1in]{geometry}
\usepackage{palatino}
\usepackage{authblk}
\usepackage{macros}

\bibliography{refs}

\makeatletter
\newcommand{\specialcell}[1]{\ifmeasuring@#1\else\omit$\displaystyle#1$\ignorespaces\fi}
\makeatother

\DeclareMathOperator{\TV}{TV}
\DeclareMathOperator{\flip}{flip}

\DeclareMathOperator{\mix}{mix}
\DeclareMathOperator{\Ent}{Ent}

\title{From Coupling to Spectral Independence and Blackbox Comparison with the Down-Up Walk}
\author{Kuikui Liu}
\affil{\small University of Washington, \textsf{liukui17@cs.washington.edu}}

\begin{document}
\maketitle
\begin{abstract}
We show that the existence of a ``good'' coupling w.r.t. Hamming distance for any local Markov chain on a discrete product space implies rapid mixing of the Glauber dynamics in a blackbox fashion. More specifically, we only require the expected distance between successive iterates under the coupling to be summable, as opposed to being one-step contractive in the worst case. Combined with recent local-to-global arguments \cite{CLV21}, we establish asymptotically optimal lower bounds on the standard and modified log-Sobolev constants for the Glauber dynamics for sampling from spin systems on bounded-degree graphs when a curvature condition \cite{Oll09} is satisfied. To achieve this, we use Stein's method for Markov chains \cite{BN19, RR19} to show that a ``good'' coupling for a local Markov chain yields strong bounds on the spectral independence of the distribution in the sense of \cite{ALO20}.



Our primary application is to sampling proper list-colorings on bounded-degree graphs. In particular, combining the coupling for the flip dynamics given by \cite{Vig00, CDMPP19} with our techniques, we show optimal $O(n\log n)$ mixing for the Glauber dynamics for sampling proper list-colorings on any bounded-degree graph with maximum degree $\Delta$ whenever the size of the color lists are at least $\wrapp{\frac{11}{6} - \epsilon}\Delta$, where $\epsilon \approx 10^{-5}$ is small constant. While $O(n^{2})$ mixing was already known before, our approach additionally yields Chernoff-type concentration bounds for Hamming Lipschitz functions in this regime, which was not known before. Our approach is markedly different from prior works establishing spectral independence for spin systems using spatial mixing \cite{ALO20, CLV20, CGSV21, FGYZ21}, which crucially is still open in this regime for proper list-colorings.

\end{abstract}

\section{Introduction}
Given a probability distribution $\mu$ on a collection of subsets of a finite universe $U$ with a fixed size $n$, one would like to generate (approximate) samples from $\mu$. This problem is widely encountered in machine learning, statistical physics, and theoretical computer science, and encompasses many problems as special cases, including distributions over bases of matroids, discrete probabilistic graphical models, etc. A popular approach used in practice is to run a Markov chain on $\supp(\mu)$ whose stationary distribution in $\mu$. The main question then becomes how quickly does the distribution of the Markov chain converge to stationarity, i.e. does it mix rapidly?

A particularly natural Markov chain known as the ``down-up walk'' (or ``high-order walk'') originally studied in the high-dimensional expander community \cite{KM17, DK17, KO18, Opp18, AL20} has recently received a lot of attention due to applications to sampling from discrete log-concave distributions \cite{ALOV19ii, CGM19, AD20, ALOVV21} and spin systems in statistical physics \cite{AL20, ALO20, CLV20, CGSV21, FGYZ21, AASV21, CLV21}. For sampling bases of matroids, the down-up walk recovers exactly the bases exchange walk first studied in \cite{FM92, MV89}, and for sampling from discrete graphical models, the down-up walk recovers exactly the classical Glauber dynamics. One of the main insights in this area is that to prove rapid mixing of the down-up walk, it suffices to look only at pairwise correlations between elements, albeit for all conditional distributions of $\mu$ (see \cref{def:specind}). One additional advantage behind this approach is that one can prove local-to-global results not just for the spectral gap \cite{KO18, AL20}, but also for the rate of entropy decay \cite{CGM19, GM20, CLV21} and even for the rate of decay for arbitrary $f$-divergences \cite{AASV21}. This has led to asymptotically optimal mixing times for many problems \cite{CGM19, ALOVV21, CLV21} as well as Chernoff-type concentration bounds for Lipschitz functions.

However, establishing sufficiently strong bounds on pairwise correlations (or, more precisely, pairwise influences; see \cref{def:specind}), remains a challenging problem. Prior works typically rely on one of three techniques: Oppenheim's trickle-down theorem \cite{Opp18, ALOV19ii}, spatial mixing (or correlation decay) \cite{ALO20, CLV20, CGSV21, FGYZ21, CLV21}, or the absence of roots for the multivariate generating polynomial of $\mu$ in a sufficiently large region of the complex plane \cite{AASV21, CLV21b}. However, there are settings, such as proper list-colorings when the number of colors is less than twice the maximum degree of the graph, where the trickle-down theorem fails, and where spatial mixing and the existence of a nice root-free region are not known.

In this work, we show that the classical technique of (path) coupling can be used to bound these pairwise correlations. In fact, we will show that the existence of a ``good'' coupling for \textbf{any} sufficiently ``local'' dynamics with stationary distribution $\mu$ implies spectral independence for $\mu$ in the sense of \cite{ALO20}, and hence, rapid mixing of the down-up walk. Hence, one can view our main result as a blackbox comparison result between any local Markov chain and the down-up walk.

As our main concrete application, we use the variable-length path coupling devised by \cite{CDMPP19} for the flip dynamics, building off work of Vigoda \cite{Vig00}, to show $O(n\log n)$ mixing of the Glauber dynamics for sampling proper list-colorings on graphs of maximum degree $\Delta \leq O(1)$ whenever the number of available colors is at least $\wrapp{\frac{11}{6} - \epsilon}\Delta$, where $\epsilon \approx 10^{-5}$ is a small constant. This mixing time is asymptotically optimal \cite{HS07}. While $O(n^{2})$ mixing was known earlier \cite{FV07} (see also \cite{Vig00, CDMPP19}) by using a spectral gap comparison argument \cite{DS93}, our approach yields optimal bounds on the rate of entropy decay as well as Chernoff-type concentration inequalities. As mentioned earlier, strong spatial mixing and the existence of sufficiently large root-free regions are not known in this regime for proper list-colorings. Along the way, we make an additional conceptual contribution by answering the natural question of if Dobrushin-type mixing conditions imply spectral independence.

\subsection{Our Contributions}
To state our blackbox comparison result, let us first define the down-up walk for sampling from distributions over homogeneous set systems. We eschew the use of much of the terminology of high-dimensional expanders so as to simplify the exposition. Let $\mu$ be a distribution over $\binom{U}{n} = \{S \subseteq U : |S| = n\}$ for a finite set $U$ and a positive integer $n \geq 1$\footnote{One can view $\mu$ as being a distribution over the edges of a hypergraph, where $U$ is the set of vertices, and the sets in $\supp(\mu)$ are hyperedges. Since $\mu$ is supported on sets of a fixed size, one can take the downwards closure of $\supp(\mu)$ to obtain a ``pure simplicial complex'' weighted by $\mu$. These are a generalization of usual graphs which are studied in geometry, topology, and combinatorics. The notion of spectral independence (\cref{def:specind}) was derived from a high-dimensional notion of ``expansion'' for simplicial complexes known as ``local spectral expansion'' first discovered by \cite{DK17, KO18, Opp18}.}. The down-up walk is described by the following two-step process. If the current state of the chain is $S^{(t)}$, then we select the next state $S^{(t+1)}$ as follows:
\begin{enumerate}
    \item Select a uniformly random element $i \in S$.
    \item Sample a set $S \in \supp(\mu)$ satisfying $S \supseteq S^{(t)} \setminus \{i\}$ with probability proportional to $\mu(S)$ and transition to $S^{(t+1)} = S$.
\end{enumerate}
As special cases, this class of Markov chains includes the bases exchange walk for matroids \cite{MV89} and the Glauber dynamics for distributions over discrete product spaces.

We also define our notion of a ``good'' coupling and locality precisely here.
\begin{definition}[Amortized Convergent Coupling]\label{def:goodcoupling}
Fix an irreducible transition probability matrix $P$ which is reversible w.r.t. a distribution $\pi$ on a finite state space $\Omega$. Further endow $\Omega$ with a metric $d(\cdot,\cdot)$. We say a coupling of two faithful copies of the chain $(X^{(t)})_{t \geq 0}, (Y^{(t)})_{t \geq 0}$ is \textbf{$C$-amortized convergent w.r.t. $d(\cdot,\cdot)$} if the following holds for all $x,y \in \Omega$:
\begin{align*}
    \sum_{t=0}^{\infty} \E_{X^{(t)}, Y^{(t)}}\wrapb{d(X^{(t)}, Y^{(t)}) \mid \substack{X^{(0)} = x \\ Y^{(0)} = y}} \leq C \cdot d(x,y).
\end{align*}
\end{definition}
\begin{definition}[Locality of Dynamics]
Fix an irreducible transition probability matrix $P$ which is reversible w.r.t. a distribution $\pi$ on a finite state space $\Omega$. Further endow $\Omega$ with a metric $d(\cdot,\cdot)$. For a positive real number $\ell > 0$, we say the dynamics $P$ is \textbf{$\ell$-local w.r.t. $d(\cdot,\cdot)$} if
\begin{align*}
    \max_{x,y \in \Omega : P(x,y) > 0} d(x,y) \leq \ell.
\end{align*}
\end{definition}
Throughout the paper, unless stated otherwise, we work with Hamming distance. With these notions in hand, we now state our blackbox comparison result.


\begin{theorem}[Blackbox Comparison with Down-Up Walk]\label{thm:blackboxcomparison}
Let $\mu$ be a distribution on $\binom{U}{n}$, where $U$ is a finite universe and $n \geq 1$ is a positive integer. For each $A \subseteq U$ with $|A| \leq n - 2$ and $A \subseteq S$ for some $S \in \supp(\mu)$, let $P_{\mu \mid A}$ be some Markov chain on $\supp(\mu \mid A)$ with stationary distribution $\mu \mid A$. Assume the family of Markov chains $\{P_{\mu \mid A}\}$ satisfy the following:
\begin{enumerate}
    \item Locality: For some $\ell \geq 0$, $P_{\mu \mid A}$ is $\ell$-local w.r.t. Hamming distance for all $A$.
    \item Good Coupling: For some $C_{n-k} > 0$, $P_{\mu \mid A}$ admits a $C_{n-k}$-amortized convergent coupling w.r.t. Hamming distance for all $k$ and $A$ with $|A| = k$.
    \item Bounded Differences Between Chains: For some $C_{n-k}' > 0$, we have the bound
    \begin{align*}
        \max_{S \in \supp(\mu \mid (A \cup i))}\wrapc{\sum_{T \neq S} \abs{P_{\mu \mid A}(S \rightarrow T) - P_{\mu \mid (A \cup i)}(S \rightarrow T)}} \leq C_{n-k}',
    \end{align*}
    for all $k$, $i$ and $A$ with $|A| = k$.
\end{enumerate}
If $\ell \cdot C_{n-k} \cdot C_{n-k}' \leq O(1)$ for all $k$, then the down-up walk has spectral gap at least $n^{-O(1)}$. If, in addition, $\mu$ is the Gibbs distribution of a spin system (see \cref{subsec:spinsystem}) on a bounded-degree graph, then the spectral gap, standard and modified log-Sobolev constants \cref{eq:mixingconstants} for the down-up walk are all $\Omega(1/n)$.
\end{theorem}
We refer the reader to \cref{app:varent} and references therein for the importance of lower bounding the spectral gap, standard and modified log-Sobolev constants, and in particular, their relation to mixing and concentration.
\begin{remark}
While initially it may seem inconvenient to first build an entire family of Markov chains, one for each conditional distribution, this is very natural for many classes of distributions, in particular those which are closed under conditioning. As we will see, in practice, it is easy to obtain bounded differences between chains with $C_{n-k}' \lesssim \frac{1}{n-k}$ simply via brute force calculation. While $C_{n-k} \gtrsim n-k$ is often unavoidable, particularly for $\ell$-local chains with $\ell \leq O(1)$, we will see that in many settings, we have $C_{n-k} \lesssim n-k$ as well. If additionally our dynamics are $\ell$-local with $\ell \leq O(1)$, then the above yields a $n^{-O(1)}$ spectral gap for the down-up walk. It will turn out that our notion of $\ell$-locality can also be relaxed; see \cref{rem:localflexible}.
\end{remark}
Our primary concrete application is to sampling proper list-colorings on graphs via the Glauber dynamics, which may be realized as a down-up walk. In this setting, we compare with another useful Markov chain known as the flip dynamics. The flip dynamics is $\ell$-local w.r.t. unweighted Hamming distance with $\ell \leq 12$, and was analyzed in \cite{Vig00}, who gave a greedy coupling which is one-step contractive whenever the number of available colors is at least $\frac{11}{6}\Delta$, implying it is $C$-amortized convergent with $C \leq O(n)$. \cite{CDMPP19} tweaked the parameters of the flip dynamics slightly while preserving locality, and further constructed a variable-length coupling which contracts by a constant factor every expected $O(n)$ steps whenever the number of available colors is at least $\wrapp{\frac{11}{6} - \epsilon}\Delta$ for a small constant $\epsilon \approx 10^{-5}$. We will show this variable-length coupling is also $C$-amortized convergent with $C \leq O(n)$, and deduce optimal mixing for list-colorings in this regime.

\begin{theorem}\label{thm:coloringspecindmain}
Let $(G,\mathcal{L})$ be a list-coloring instance where $G=(V,E)$ is a graph of maximum degree $\Delta \leq O(1)$ and $\mathcal{L} = (L(v))_{v \in V}$ is a collection of color lists. Then for some absolute constant $\epsilon \approx 10^{-5}$, if $\abs{L(v)} \geq \wrapp{\frac{11}{6} - \epsilon}\Delta$ for all $v \in V$, then the uniform distribution over proper list-colorings for $(G,\mathcal{L})$ is $(\eta_{0},\dots,\eta_{n-2})$-spectrally independent where $\eta_{k} \leq O(1)$ for all $k$. Furthermore, the spectral gap, standard and modified log-Sobolev constants \cref{eq:mixingconstants} for the Glauber dynamics are all $\Omega(1/n)$, and the mixing time is $O(n\log n)$.
\end{theorem}
\begin{remark}
Our running time dependence on $\Delta$ is roughly $\Delta^{\Delta^{c}}$ for a mild constant $c$, which is rather poor. The main bottleneck in improving this dependence lies in the local-to-global result of \cite{CLV21}, although our spectral independence bound, which depends polynomially on $\Delta$, can also be significantly improved.
\end{remark}

To prove \cref{thm:blackboxcomparison}, we leverage recent local-to-global results \cite{AL20, CLV21} (see \cref{thm:localtoglobalspectralgap,thm:localtoglobalspinoptimal} for formal statements), which show that if one has sufficiently strong upper bounds on the total pairwise correlation $\sum_{j \in U} \abs{\Pr_{S \sim \mu}[j \in S] - \Pr_{S \sim \mu \mid i}[j \in S]}$, then one can deduce rapid mixing for the down-up walk \cite{KM17, DK17, KO18, Opp18}. To upper bound these correlations, we considerably generalize a result simultaneously due to \cite{BN19, RR19}, which was discovered in the context of bounding the Wasserstein $1$-distance between Ising models, or more generally, two measures on the discrete hypercube $\{-1,+1\}^{n}$. More specifically, we extend their results in several different directions:
\begin{enumerate}
    \item We replace the Glauber dynamics by any local dynamics.
    \item We allow the dynamics to admit a coupling which in a sense ``contracts on average'', as opposed to a step-wise contraction in the worst-case.
\end{enumerate}
\begin{theorem}\label{thm:specindlocalamortized}
Let $\mu$ be a distribution on $\binom{U}{n}$, where $U$ is some finite universe and $n \geq 1$ is a positive integer. Fix an arbitrary $i \in U$. Let $P_{\mu}$ (resp. $P_{\mu \mid i}$) be the transition kernel of any irreducible Markov chain on $\supp(\mu)$ (resp. $\supp(\mu \mid i)$) which is reversible w.r.t. $\mu$ (resp. $\supp(\mu \mid i)$). Suppose that $P_{\mu}$ is $\ell$-local and admits a $C$-amortized convergent coupling, both w.r.t. the Hamming metric $d_{H}(\cdot,\cdot)$. Then we have the bound
\begin{align*}
    &\sum_{j \in U} \abs{\Pr_{S \sim \mu}[j \in S] - \Pr_{S \sim \mu \mid i}[j \in S]} \leq C \cdot \ell \cdot \max_{S \in \supp(\mu \mid i)} \wrapc{\sum_{T \neq S} \abs{P_{\mu}(S \rightarrow T) - P_{\mu \mid i}(S \rightarrow T)}}.
\end{align*}
\end{theorem}

\subsection{Main Technical Result}
We now state our main technical result, which provides the most general bound on the difference between marginals of two distributions $\mu,\nu$. We immediately use it to deduce \cref{thm:specindlocalamortized}.

\begin{theorem}[Main Technical]\label{thm:maintechnical}
Let $\mu,\nu$ be any two distributions on $2^{U}$ for a finite set $U$ with $\supp(\nu) \subseteq \supp(\mu)$, where $U$ is a finite universe and $n \geq 1$ is a positive integer. 
Further, let $P_{\mu}$ (resp. $P_{\nu}$) be the transition kernel of any Markov chain on $\supp(\mu)$ (resp. $\supp(\nu)$) with stationary distribution $\mu$ (resp. $\nu$). Assume $P_{\mu}$ is irreducible and reversible w.r.t. $\mu$. Then we may bound both $\sum_{j \in U} \abs{\Pr_{S \sim \mu}[j \in S] - \Pr_{S \sim \nu}[j \in S]}$ and the $1$-Wasserstein distance $W_{1}(\mu,\nu)$ (see \cref{def:wasserstein}) by the following quantity:
\begin{align*}
    \E_{S \sim \nu} \wrapb{\sum_{T \neq S} \abs{P_{\mu}(S \rightarrow T) - P_{\nu}(S \rightarrow T)} \cdot \sum_{t=0}^{\infty} \E_{X^{(t)},Y^{(t)}} \wrapb{d_{H}(X^{(t)}, Y^{(t)}) \mid \substack{X^{(0)} = S \\ Y^{(0)} = T}}},
\end{align*}
where $(X^{(t)},Y^{(t)})_{t=0}^{\infty}$ is a coupling of the Markov chain $P_{\mu}$.
\end{theorem}
\begin{remark}
The technical condition $\supp(\nu) \subseteq \supp(\mu)$ is just for convenience, as it ensures the transition probability $P_{\mu}(S \rightarrow T)$ also makes sense when $S \sim \nu$. This assumption is certainly satisfied in our application where $\nu$ is a conditional distribution of $\mu$.
\end{remark}
\begin{proof}[Proof of \cref{thm:specindlocalamortized}]
We use \cref{thm:maintechnical} with $\nu = \mu \mid i$ to obtain the upper bound
\begin{align*}
    &\E_{S \sim \mu \mid i} \wrapb{\sum_{T \neq S} \abs{P_{\mu}(S \rightarrow T) - P_{\mu \mid i}(S \rightarrow T)} \cdot \sum_{t=0}^{\infty} \E_{X^{(t)},Y^{(t)}} \wrapb{d_{H}(X^{(t)}, Y^{(t)}) \mid \substack{X^{(0)} = S \\ Y^{(0)} = T}}} \\
    &\leq \max_{S \in \supp(\mu \mid i)}\wrapc{\sum_{T \neq S} \abs{P_{\mu}(S \rightarrow T) - P_{\mu \mid i}(S \rightarrow T)}}\\
    &\quad\quad\quad\quad\quad \cdot\underset{(*)}{\underbrace{\E_{S \sim \mu \mid i}\wrapb{\max_{T : P_{\mu}(S \rightarrow T) > 0} \sum_{t=0}^{\infty} \E_{X^{(t)},Y^{(t)}} \wrapb{d_{H}(X^{(t)}, Y^{(t)}) \mid \substack{X^{(0)} = S \\ Y^{(0)} = T}}}}}.
\end{align*}
It suffices to bound $(*)$ by $C \cdot \ell$. Since $P_{\mu}$ admits a $C$-amortized convergent coupling, we have that
\begin{align*}
    \sum_{t=0}^{\infty} \E_{X^{(t)},Y^{(t)}} \wrapb{d_{H}(X^{(t)}, Y^{(t)}) \mid \substack{X^{(0)} = S \\ Y^{(0)} = T}} \leq C \cdot d_{H}(S,T).
\end{align*}
Hence,
\begin{align*}
    (*) \leq C \cdot \E_{S \sim \mu \mid i}\wrapb{\max_{T : P_{\mu}(S \rightarrow T) > 0} d_{H}(S,T)} \leq C \cdot \ell.
\end{align*}
\end{proof}
\begin{remark}\label{rem:localflexible}
One can see from the proof that we only needed that
\begin{align*}
    \E_{S \sim \mu\mid i}\wrapb{\max_{T : P_{\mu}(S \rightarrow T)} d_{H}(S,T)} \leq \ell,
\end{align*}
as opposed to the stronger notion of $\ell$-locality, where we have $\max_{S,T : P_{\mu}(S \rightarrow T)} d_{H}(S,T) \leq \ell$. Thus, in some sense, we only need the dynamics to make local moves ``on average''. We leave it to future work to exploit this additional flexibility.
\end{remark}

\subsection{Independent Work}
The results we obtain here were also independently discovered in \cite{BCCPSV21}.
\subsection{Acknowledgements}
The author is supported by NSF grants CCF-1552097, CCF-1907845. The author would like to thank their advisor Shayan Oveis Gharan for comments on a preliminary draft of this paper. We also thank Nima Anari and Pierre Youssef for informing us that a conjecture posed in a preliminary draft of the paper was already known to be false. We finally thank the anonymous reviewers for delivering valuable feedback on this paper.

\subsection{Organization of the Paper}
We state preliminaries on spin systems, spectral independence, etc. in \cref{sec:prelim}. We then move to the proof of our main technical result (\cref{thm:maintechnical}) in \cref{sec:stein}. In \cref{sec:productspaces}, we apply our techniques to distributions on discrete product spaces. In \cref{sec:listcolorings}, we combine our techniques with couplings constructed in prior works to obtain spectral independence for proper list-colorings.

\section{Preliminaries}\label{sec:prelim}
For a positive integer $n \geq 1$, we write $[n] = \{1,\dots,n\}$. For a distribution $\mu$ on some finite state space $\Omega$, we write $\supp(\mu) = \{x \in \Omega : \mu(x) > 0\}$ for the support of $\mu$. For a matrix $A$, we write $\norm{A}_{\infty} = \max_{i} \sum_{j} |A(i,j)|$ for the maximum absolute row sum, and if $A$ has real eigenvalues, we write $\lambda_{\max}(A)$ for the largest eigenvalue of $A$.

Throughout, we write $G=(V,E)$ for an undirected graph, and we will write $\Delta$ for the maximum degree of $G$. For a finite universe $U$ and $S \subseteq U$, we write $\mathbb{I}_{S}$ for the $\{0,1\}$-indicator function of $S$; for an element $j \in U$, we write $\mathbb{I}_{j}$ as opposed to $\mathbb{I}_{\{j\}}$. If $\mu$ is a distribution over $\binom{U}{n}$ and $S \subseteq U$, then we write $\mu \mid S$ for the conditional distribution of $\mu$ on $\binom{U \setminus S}{n - |S|}$, where $(\mu \mid S)(T) \propto \mu(S \cup T)$ whenever $S \cup T \in \supp(\mu), S \cap T = \emptyset$, and $(\mu \mid S)(T) = 0$ otherwise.


We will measure convergence of our Markov chains using total variation distance, defined as
\begin{align*}
    d_{\TV}(\mu,\nu) = \frac{1}{2} \sum_{x} \abs{\mu(x) - \nu(x)} = \sup_{S \subseteq \Omega} \abs{\mu(S) - \nu(S)}
\end{align*}
for two distributions $\mu,\nu$ on a common state space $\Omega$. We define the $\epsilon$-mixing time of a Markov chain $P$ on a state space $\Omega$ with stationary distribution $\pi$ as
\begin{align*}
    t_{\mix}(\epsilon) \overset{\defin}{=} \max_{x \in \Omega} \min\{t \geq 0 : d_{\TV}(\mathbb{I}_{x}P^{t}, \pi) \leq \epsilon\}.
\end{align*}
The mixing time of the chain is defined as $t_{\mix}(1/4)$. For the reader's convenience, in \cref{app:varent}, we record the relation between mixing, spectral gap, modified and standard log-Sobolev constants.
Finally, we also define the $1$-Wasserstein distance.
\begin{definition}[$1$-Wasserstein Distance]\label{def:wasserstein}
Given two probability measures $\mu,\nu$ on a common state space $\Omega$ endowed with a metric $d(\cdot,\cdot)$, we define the $1$-Wasserstein distance $W_{1}(\mu,\nu)$ w.r.t. $d(\cdot,\cdot)$ by
\begin{align*}
    W_{1}(\mu,\nu) = \sup_{f}\abs{\E_{\mu}f - \E_{\nu} f},
\end{align*}
where the supremum is over functions $f:\Omega \rightarrow \R$ which are $1$-Lipschitz w.r.t. $d(\cdot,\cdot)$ (i.e. $\abs{f(x) - f(y)} \leq d(x,y)$ for all $x,y \in \Omega$).
\end{definition}
\begin{remark}
By Kantorovich duality, one may equivalently define the $1$-Wasserstein distance as
\begin{align*}
    W_{1}(\mu,\nu) = \inf_{\gamma} \E_{(x,y) \sim \gamma}[d(x,y)],
\end{align*}
where the infimum is overall couplings $\gamma$ of $\mu,\nu$ on $\Omega \times \Omega$.
\end{remark}

\subsection{Spin Systems}\label{subsec:spinsystem}
Fix an undirected graph $G=(V,E)$, and a positive integer $q \geq 2$. We view $[q]$ as a collection of possible ``spin assignments'' for the vertices of $G$. We also fix a symmetric nonnegative matrix $A \in \R_{\geq0}^{q \times q}$ of ``edge interaction activities'' and a positive vector $h \in \R_{>0}^{q}$ of ``external fields''. The Gibbs distribution of the spin system on $G=(V,E)$ with parameters $A,h$ is the distribution $\mu = \mu_{G,A,h}$ over configurations $\sigma:V \rightarrow [q]$ given by
\begin{align*}
    \mu(\sigma) \propto \prod_{\{u,v\} \in E} A(\sigma(u), \sigma(v)) \prod_{v \in V} h(\sigma(v)),
\end{align*}
where the constant of proportionality is the partition function of the system, given by
\begin{align*}
    Z_{G}(A,h) = \sum_{\sigma :V \rightarrow [q]} \prod_{\{u,v\} \in E} A(\sigma(u), \sigma(v)) \prod_{v \in V} h(\sigma(v)).
\end{align*}
Many classical models in statistical physics as well as distributions over often-studied combinatorial objects on graphs may be found as special cases:
\begin{enumerate}
    \item Ising Model of Magnetism (Cuts): $A = \begin{bmatrix} e^{\beta} & 1 \\ 1 & e^{\beta} \end{bmatrix}$ and $h > 0$ is a magnetic field
    \item Hardcore Gas Model (Independent Sets): $A = \begin{bmatrix} 0 & 1 \\ 1 & 1 \end{bmatrix}$ and $h = \lambda \mathbf{1}$ where $\lambda > 0$
    \item Monomer-Dimer Model (Matchings): $A = \begin{bmatrix} 0 & 1 \\ 1 & 1 \end{bmatrix}$ and $h = \lambda \mathbf{1}$ where $\lambda > 0$ (the same parameters as the hardcore model), with the restriction that $G$ is a line graph
    \item Zero-Temperature Antiferromagnetic Potts Model (Proper Colorings): $A = J_{q} - I_{q}$ and $h = \mathbf{1}$, where $J_{q}$ is the $q \times q$ all-$1$s matrix, and $I_{q}$ is the $q \times q$ identity matrix
\end{enumerate}
We call a configuration $\sigma:V \rightarrow [q]$ feasible if $\mu(\sigma) > 0$. For instance, if $A$ has all positive entries, then all configurations $\sigma:V \rightarrow [q]$ are feasible. We call a partial configuration $\xi :S \rightarrow [q]$, where $S \subseteq V$ is a subset of vertices, a boundary condition. For such a boundary condition, we write $\mu \mid \xi$ for the conditional Gibbs distribution on $V \setminus S$ given by taking $\mu$ and conditioning on the event that the sampled $\sigma \sim \mu$ satisfies $\sigma(v) = \xi(v)$ for all $v \in S$.

\subsection{Discrete Product Spaces and Homogeneous Set Systems}\label{subsec:productspacehomogprelim}
Fix a collection of finite sets $(\Omega(v))_{v \in V}$, where $V$ is some finite index set with $|V| = n$, and consider a measure $\mu$ on the product space $\prod_{v \in V} \Omega(v)$. For instance, if $\Omega(v) = \{-1,+1\}$ for each $v \in V$, then $\mu$ is just a measure on the discrete hypercube $\{-1,+1\}^{V}$. An important subclass of examples which we will discuss at length include discrete probabilistic graphical models, where the index set $V$ is the set of vertices of a (hyper)graph, and the measure $\mu$ designed in such a way that the (hyper)edges represent local interactions between vertices of the model; see \cref{subsec:spinsystem} for more details.

As done in \cite{ALO20, CLV20, CGSV21, FGYZ21}, we view $\mu$ as a measure on $\binom{U}{n}$ where
\begin{align*}
    U = \{(v, \omega(v)) : v \in V, \omega(v) \in \Omega(v)\}.
\end{align*}
Note the usual Hamming distance $d_{H}(\cdot,\cdot)$ on $\binom{U}{n}$ is twice the usual Hamming distance typically associated with a discrete product space.

We will often write a single vertex-assignment pair $(v,c)$, where $v \in V$ and $c \in \Omega(v)$, as simply $vc$. Here, each configuration $\sigma \in \prod_{v \in V} \Omega(v)$ corresponds to the set $\{(v, \sigma(v)) : v \in V\}$. In this setting, the down-up walk is precisely the \textbf{Glauber dynamics} (or \textbf{Gibbs sampler}) for sampling from $\mu$. For each configuration $\sigma \in \prod_{v \in V} \Omega(v)$, we transition to the next configuration by the following process:
\begin{enumerate}
    \item Select a uniformly random coordinate $v \in V$.
    \item Resample $\sigma(v)$ according to $\mu$ conditioned on $\sigma_{-v}$.
\end{enumerate}
Let us make this more concrete. For each $\sigma \in \prod_{v \in V} \Omega(v)$ and $v \in V$, we write $\sigma_{-v}$ for the partial subconfiguration of $\sigma$ which only excludes $\sigma(v)$. For $c \in \Omega(v)$, we also write $\sigma_{vc}$ for the configuration obtained by flipping the coordinate of $v$ from $\sigma(v)$ to $c$. We may then write $\mu^{v}(\cdot \mid \sigma_{-v})$ for the marginal distribution of $\sigma(v)$ under $\mu$ conditioned on $\sigma_{-v}$. The transition kernel of the Glauber dynamics may then be written as
\begin{align*}
    P_{\mu}(\sigma \rightarrow \sigma_{vc}) = \frac{1}{n} \cdot \mu^{v}(c \mid \sigma_{-v}).
\end{align*}

\subsection{Spectral Independence and The Down-Up Walk}

Here, we formalize spectral independence and its connection with rapid mixing of the down-up walk. Throughout the paper, we will assume the following connectivity/nondegeneracy condition. Assuming \cref{thm:specindlocalamortized} holds, we will also give a proof of \cref{thm:blackboxcomparison}.

\textbf{Assumption:} The down-up walk for $\mu$ and all of its conditional distributions is connected. In the context of spin systems, this condition is guaranteed by ``total connectivity'' of the system parameters $(A,h)$ \cite{CLV21}. For instance, this is satisfied by all ``soft-constraint'' models (i.e. those with $A > 0$), and many ``hard-constraint'' models such as the hardcore model and the uniform distribution over proper colorings when $q \geq \Delta + 2$.

Let us now formalize spectral independence.
\begin{definition}[Pairwise Influence and Spectral Independence \cite{ALO20}]\label{def:specind}
Fix a finite universe $U$ and a positive integer $n \geq 1$. Fix a distribution $\mu$ on $\binom{U}{n} = \{S \subseteq U : |S| = n\}$. We define the \textbf{pairwise influence} of an element $i$ on another element $j$ by
\begin{align*}
    \mathcal{I}_{\mu}(i \rightarrow j) \overset{\defin}{=} \Pr_{S \sim \mu}[j \in S \mid i \in S] - \Pr_{S \sim \mu}[j \in S].
\end{align*}
We write $\mathcal{I}_{\mu} \in \R^{U \times U}$ defined by $\mathcal{I}_{\mu}(i, j) = \mathcal{I}_{\mu}(i \rightarrow j)$ for the \textbf{pairwise influence matrix of $\mu$}. We say the distribution $\mu$ is \textbf{$\eta$-spectrally independent} if $\lambda_{\max}(\mathcal{I}_{\mu}) \leq \eta + 1$. We say the distribution $\mu$ is \textbf{$(\eta_{0},\dots,\eta_{n-2})$-spectrally independent} if $\mu$ is $\eta_{0}$-spectrally independent, $\mu \mid i$ is $\eta_{1}$-spectrally independent for all $i \in U$, and so on.
\end{definition}
Often in practice, and in this paper, instead of bounding $\lambda_{\max}(\mathcal{I}_{\mu})$, we will bound
\begin{align*}
    \norm{\mathcal{I}_{\mu}}_{\infty} = \max_{i \in U} \sum_{j \in U} \abs{\mathcal{I}_{\mu}(i \rightarrow j)},
\end{align*}
which is sufficient since it is well-known that $\lambda_{\max}(A) \leq \norm{A}_{\infty}$ for any matrix $A$ with real eigenvalues. The main usefulness of spectral independence is that it implies rapid mixing of the down-up walk, while only requiring bounds on pairwise correlations. We state the main local-to-global results most relevant to us here. In the most general setting, we may deduce an inverse polynomial spectral gap from sufficiently strong spectral independence \cite{DK17, KO18, AL20}.
\begin{theorem}[\cite{AL20}, \cite{ALO20}]\label{thm:localtoglobalspectralgap}
Let $U$ be a finite universe, and $n \geq 1$ a positive integer. Let $\mu$ be a distribution on $\binom{U}{n}$ which is $(\eta_{0},\dots,\eta_{n-2})$-spectrally independence. Then the down-up walk on $\binom{U}{n}$ for sampling from $\mu$ has spectral gap at least
\begin{align*}
    \frac{1}{n} \prod_{k=0}^{n-2} \wrapp{1 - \frac{\eta_{k}}{n - k - 1}}.
\end{align*}
In particular, if $\eta_{k} \leq O(1)$ for all $k=0,\dots,n-2$, then we have $n^{O(1)}$-mixing of the down-up walk.
\end{theorem}
Subject to a certain mild technical condition on the marginals of the distribution $\mu$, one can transfer spectral independence bounds to ``local entropy decay'' bounds, and then employ versions of the local-to-global result for entropy decay \cite{GM20, CLV21, AASV21}. In the setting of spin systems, one can further take advantage of the bounded-degree assumption to obtain $O(n\log n)$ mixing time upper bounds \cite{CLV21}, which are asymptotically optimal \cite{HS07}. We state the current state-of-the-art for spin systems on bounded-degree graphs here, as we will need it in our application to proper list-colorings.
\begin{theorem}[\cite{CLV21}]\label{thm:localtoglobalspinoptimal}
Let $(A,h)$ be the parameters of a spin system, and let $G=(V,E)$ be a graph with maximum degree at most $\Delta \leq O(1)$. If the Gibbs distribution $\mu = \mu_{G,A,h}$ is $(\eta_{0},\dots,\eta_{n-2})$-spectrally independent where $\eta_{k} \leq O(1)$ for all $k$, then both the standard and modified log-Sobolev constants of the Glauber dynamics (i.e. the down-up walk) for sampling from $\mu$ are at least $\Omega(1/n)$.
\end{theorem}
We conclude this section with a proof of \cref{thm:blackboxcomparison}.
\begin{proof}[Proof of \cref{thm:blackboxcomparison}]
By \cref{thm:localtoglobalspectralgap,thm:localtoglobalspinoptimal}, it suffices to establish $(\eta_{0},\dots,\eta_{n-2})$-spectral independence for $\eta_{k} \leq O(1)$ for all $k$. Fix an arbitrary $A \subseteq U$ with $|A| = k \leq n-2$ and $A \subseteq S$ for some $S \in \supp(\mu)$. With the locality and coupling assumptions, \cref{thm:specindlocalamortized} shows that for each $i \in U$, the absolute row sum of $\mathcal{I}_{\mu \mid A}$ for row $i$ is upper bounded by
\begin{align*}
    \ell \cdot C_{n-k} \cdot \max_{S \in \supp(\mu \mid i)} \wrapc{\sum_{T \neq S} \abs{P_{\mu}(S \rightarrow T) - P_{\mu \mid i}(S \rightarrow T)}}.
\end{align*}
Bounded differences between chains then yields the upper bound $\lambda_{\max}(\mathcal{I}_{\mu \mid A}) \leq \norm{\mathcal{I}_{\mu \mid A}}_{\infty} \leq \ell \cdot C_{n-k} \cdot C_{n-k}'$, which is $O(1)$ by assumption. As this holds for all such $A$, it follows that $\eta_{k} \leq O(1)$.
\end{proof}

\section{Stein's Method for Markov Chains}\label{sec:stein}
Our goal in this section is to prove \cref{thm:maintechnical}. We follow \cite{BN19, RR19}, using what is known as \textbf{Stein's method for Markov chains}. Historically, Stein's method \cite{Ste72} was developed as a method to bound distances between probability measures, with the primary motivation being to prove quantitative central limit theorems. \cite{BN19, RR19} adapted this method to bound the distance between two probability measures $\mu,\nu$ on the discrete hypercube $\{-1,+1\}^{n}$ assuming the Glauber dynamics of either measure admits a contractive coupling. Our main intuition lies in viewing spectral independence (see \cref{def:specind}) as a measure of distance between different conditionings of the same distribution. Thus, one can try to apply this method to bound the spectral independence of a distribution. Let us now elucidate this method.

For a fixed function $f:\Omega \rightarrow \R$, we will construct an auxiliary function $h:\Omega \rightarrow \R$ which satisfies the \textbf{Poisson equation}
\begin{align*}
    h - P_{\mu}h = f - \E_{\mu}f.
\end{align*}
Questions concerning $\E_{\mu}f$ may then be studied by looking at $P_{\mu}h$. The following lemma constructs $h$ more explicitly.
\begin{lemma}[see Lemma 2.1 \cite{BN19}, Lemma 2.3 \cite{RR19}]\label{lem:poissonsolution}
Fix an irreducible transition probability matrix $P$ which is reversible w.r.t. a distribution $\pi$ on a finite state space $\Omega$. Let $(X^{(t)})_{t=0}^{\infty}$ be the Markov chain generated by $P$, and for a fixed function $f:\Omega \rightarrow \R$, define $h:\Omega \rightarrow \R$ by
\begin{align*}
    h(x) = \sum_{t=0}^{\infty} \E\wrapb{f(X^{(t)}) - \E_{\pi}f \mid X^{(0)} = x}.
\end{align*}
Then $h$ is well-defined as a function, and further satisfies the Poisson equation
\begin{align*}
    h - Ph = f - \E_{\pi}f.
\end{align*}
\end{lemma}
With this lemma in hand, we can immediately prove \cref{thm:maintechnical}.
\begin{proof}[Proof of \cref{thm:maintechnical}]
Fix a function $f:2^{U} \rightarrow \R$, and let $h$ be the solution to the Poisson equation $h - P_{\mu}h = f - \E_{\mu}f$ given in \cref{lem:poissonsolution}.
Then since $\nu$ is stationary w.r.t. $P_{\nu}$, we have $\E_{\nu}P_{\nu}h = \E_{\nu}h$, so that using the Poisson equation yields
\begin{align*}
    \E_{\nu} (P_{\nu} - P_{\mu})h &= \E_{\nu} h - \E_{\nu}\wrapb{h - f + \E_{\mu}f} = \E_{\nu}f - \E_{\mu}f.
\end{align*}
Hence, by the Triangle Inequality, we have that $\abs{\E_{\mu}f - \E_{\nu}f} \leq \E_{\nu} \abs{(P_{\nu} - P_{\mu})h}$.

Now, let us bound $\abs{(P_{\nu} - P_{\mu})h}$ entrywise. For each $S \in \supp(\nu)$, using that $P_{\mu}(S \rightarrow S) = 1 - \sum_{T \neq S} P_{\mu}(S \rightarrow T)$ (and analogously for $P_{\nu}$),
\begin{align*}
    (P_{\nu} - P_{\mu})h(S) &= \sum_{T} (P_{\nu}(S \rightarrow T) - P_{\mu}(S \rightarrow T)) \cdot h(T) \\
    &= \sum_{T \neq S} (P_{\nu}(S \rightarrow T) - P_{\mu}(S \rightarrow T)) \cdot (h(T) - h(S)) \\
    &= \sum_{T \neq S} (P_{\nu}(S \rightarrow T) - P_{\mu}(S \rightarrow T)) \cdot \sum_{t=0}^{\infty} \E_{X^{(t)},Y^{(t)}}\wrapb{f(Y^{(t)}) - f(X^{(t)}) \mid \substack{X^{(0)} = S \\ Y^{(0)} = T}}. \tag{\cref{lem:poissonsolution}}
\end{align*}
It follows by the Triangle Inequality that
\begin{align}\label{eq:entrywisebound}
    \abs{(P_{\nu} - P_{\mu})h(S)} &\leq \sum_{T \neq S} \abs{P_{\mu}(S \rightarrow T) - P_{\nu}(S \rightarrow T)} \cdot \sum_{t=0}^{\infty} \E_{X^{(t)},Y^{(t)}} \wrapb{\abs{f(X^{(t)}) - f(Y^{(t)})} \mid \substack{X^{(0)} = S \\ Y^{(0)} = T}}.
\end{align}
Taking expectations w.r.t. $\nu$ finally yields a bound on $\abs{\E_{\mu}f - \E_{\nu}f}$. The bound on the $1$-Wasserstein distance follows immediately by taking $f$ to be an arbitrary function which is $1$-Lipschitz the metric $d_{H}(\cdot,\cdot)$. To obtain the bound on the total difference between marginals $\sum_{j \in U} \abs{\Pr_{S \sim \mu}[j \in S] - \Pr_{S \sim \nu}[j \in S]}$, we apply the above inequality to $f = \mathbb{I}_{j}$ for each $j \in U$ and sum over all $j \in U$, noting that $d_{H}(S,T) = \sum_{j \in U} \abs{\mathbb{I}_{j}(S) - \mathbb{I}_{j}(T)}$ and $\E_{\mu}f = \E_{\mu} \mathbb{I}_{j} = \Pr_{S \sim \mu}[j \in S]$ (and analogously for $\nu$).
\end{proof}

\section{Discrete Ricci Curvature on Product Spaces}\label{sec:productspaces}
In this section, we discuss applications of our results to general distributions on discrete product spaces. We show that the existence of a contractive coupling w.r.t. Hamming distance for the Glauber dynamics implies $O(1)$-spectral independence. Such a condition is known as a discrete Ricci curvature condition for the dynamics in the sense of \cite{Oll09}. This also shows that the Dobrushin uniqueness condition implies $O(1)$-spectral independence. When combined with the local-to-global result of \cite{CLV21}, we resolve an unpublished conjecture due Peres and Tetali for spin systems on bounded-degree graphs; see \cite{ELL17} and references therein for recent progress on this conjecture on general graphs. We also give an alternative proof of the $\Omega(1/n)$ lower bound on the standard and modified log-Sobolev constants of the Glauber dynamics in this setting when a Dobrushin-type condition is satisfied, recovering a result of \cite{Mar19}. 

Classical work on Dobrushin-type conditions \cite{Dob70, DS85i, DS85ii, DS87, Hay06, DGJ09} yield relatively simple and direct criteria for rapid mixing of the Glauber dynamics \cite{BD97i, BD97ii}. The main idea here is intuitively similar to that of spectral independence (although the notion of Dobrushin influence here historically precedes spectral independence): so long as some measure of ``total influence'' is small, then $\mu$ is close in some sense to a product distribution, for which rapid mixing holds. However, prior to our work, the precise relationship between Dobrushin influence and the notion of pairwise influence used in spectral independence was unclear. This is an additional conceptual contribution of our work.


\begin{definition}[Discrete Ricci Curvature \cite{Oll09}]\label{def:contractive}
Fix an irreducible transition probability matrix $P$ which is reversible w.r.t. a distribution $\pi$ on a finite state space $\Omega$. Further, endow $\Omega$ with a metric $d(\cdot,\cdot)$. We define the \textbf{discrete Ricci curvature} of the Markov chain $P$ w.r.t. the metric space $(\Omega,d)$ by
\begin{align*}
    \alpha = \inf_{x,y \in \Omega : x \neq y} \wrapc{1 - \frac{W_{1}(P(x \rightarrow \cdot), P(y \rightarrow \cdot))}{d(x,y)}},
\end{align*}
where $W_{1}(\cdot,\cdot)$ is again the $1$-Wasserstein distance w.r.t. $d(\cdot,\cdot)$. In other words, for every pair $x,y \in \Omega$, there is a coupling of the transitions $P(x \rightarrow \cdot), P(y \rightarrow \cdot)$ such that the expected distance $d(\cdot,\cdot)$ under the coupling contracts by a $(1-\alpha)$-multiplicative factor. In this case, we will say $P$ admits a \textbf{$(1-\alpha)$-contractive coupling w.r.t. $d(\cdot,\cdot)$}.
\end{definition}
\begin{fact}\label{fact:contractiveimpliesgood}
Suppose $P$ admits a $(1-\alpha)$-contractive coupling w.r.t. $d(\cdot,\cdot)$. Then this coupling is $C$-amortized convergent with $C = \frac{1}{\alpha}$.
\end{fact}

The following is an immediate application of \cref{thm:specindlocalamortized}, and yields a positive resolution to the Peres-Tetali conjecture for spin systems on bounded-degree graphs.

\begin{theorem}[Curvature Implies Spectral Independence on Product Spaces]\label{thm:curvaturespecind}
Let $\mu$ be a measure on a discrete product space $\Omega = \prod_{v \in V} \Omega(v)$, where $V$ is a finite index set and $\Omega(v)$ is finite for all $v \in V$. Endow $\Omega$ with the Hamming metric $d_{H}(\cdot,\cdot)$, and let $\alpha$ be the discrete Ricci curvature of the Glauber dynamics w.r.t. $(\Omega,d_{H})$. Then, the distribution is $(\eta_{0},\dots,\eta_{n-2})$-spectrally independent where $\eta_{k} \leq \frac{4}{\alpha n} - 1$ for all $k$. In particular, if $\alpha \geq \Omega(1/n)$, then the Glauber dynamics has spectral gap $n^{-O(1)}$. If additionally the measure $\mu$ is the Gibbs distribution of a spin system on a bounded-degree graph, then the spectral gap, standard and modified log-Sobolev constants for the Glauber dynamics are all $\Omega(1/n)$.
\end{theorem}
Note that since the Glauber dynamics only updates the assignment to a single $v \in V$ in each step, it must be that $\alpha \leq O(1/n)$.
\begin{proof}
We show that $\eta_{0} \leq \frac{4}{\alpha n} - 1$. The bound $\eta_{k} \leq \frac{4}{\alpha n} - 1$ follows by the same argument by instead considering the Glauber dynamics for the conditional distributions $\mu \mid A$ of $\mu$. Because the Glauber dynamics only updates at most one coordinate in each step, it is $2$-local w.r.t. $d_{H}(\cdot,\cdot)$. By \cref{fact:contractiveimpliesgood}, we also have there is a $C$-amortized convergent coupling with $C = \frac{1}{\alpha}$. It follows from \cref{thm:specindlocalamortized} that
\begin{align*}
    &\sum_{v \in V} \sum_{c' \in \Omega(v)} \abs{\Pr_{\sigma \sim \mu}[\sigma(v) = c' \mid \sigma(u) = c] - \Pr_{\sigma \sim \mu}[\sigma(v) = c']} \leq \frac{2}{\alpha} \max_{\sigma \in \supp(\mu \mid uc)}\sum_{\tau \neq \sigma} \abs{P_{\mu}(\sigma \rightarrow \tau) - P_{\mu \mid uc}(\sigma \rightarrow \tau)}.
\end{align*}
Now, by the definition of the Glauber dynamics, for each $\sigma \in \supp(\mu \mid uc)$, we have
\begin{align*}
    &\sum_{\tau \neq \sigma} \abs{P_{\mu}(\sigma \rightarrow \tau) - P_{\mu \mid uc}(\sigma \rightarrow \tau)} \\
    &= \sum_{v \in V} \sum_{c' \in L(v) : c' \neq \sigma(v)} \abs{\frac{1}{n}\mu^{v}(c' \mid \sigma_{-v}) - \frac{1}{n-1}\mu_{uc}^{v}(c' \mid \sigma_{-v})} \\
    &= \sum_{v \in V : v \neq u} \sum_{c' \in L(v) : c' \neq \sigma(v)} \wrapp{\frac{1}{n-1} - \frac{1}{n}} \mu^{v}(c' \mid \sigma_{-v}) + \sum_{c' \in L(u) : c' \neq c} \frac{1}{n}\mu^{v}(c' \mid \sigma_{-v}) \\
    &\leq \frac{2}{n}.
\end{align*}
The claim for the spectral gap in the case $\alpha \geq \Omega(1/n)$ follows by combining with \cref{thm:localtoglobalspectralgap}. The final claim for spin systems on bounded-degree graphs follows by combining with \cref{thm:localtoglobalspinoptimal}.
\end{proof}

\subsection{Dobrushin Uniqueness and Spectral Independence}
We now use \cref{thm:curvaturespecind} to show that Dobrushin's uniqueness condition implies spectral independence.
\begin{definition}[Dobrushin Influence]
Fix a probability measure on a finite product space $\prod_{v \in V} \Omega(v)$, where $V$ is a finite indexing set. For each $u \in V$, let $D_{u}$ be the collection of pairs $\tau,\sigma \in \prod_{v \in V} \Omega(v)$ such that $\tau_{-u} = \sigma_{-u}$ while $\tau(u) \neq \sigma(u)$. For distinct $u,v \in V$, we may then define the \textbf{Dobrushin influence} of $u$ on $v$ by
\begin{align*}
    \rho_{\mu}(u \rightarrow v) = \max_{(\tau,\sigma) \in D_{u}} d_{\TV}(\mu^{v}(\cdot \mid \tau_{-v}), \mu^{v}(\cdot \mid \sigma_{-v})).
\end{align*}
We write $\rho_{\mu} = (\rho_{\mu}(u \rightarrow v))_{u,v} \in \R^{V \times V}$ for the \textbf{Dobrushin influence matrix}. We say the distribution $\mu$ satisfies the \textbf{Dobrushin uniqueness condition} if
\begin{align*}
    \norm{\rho_{\mu}}_{1} \overset{\defin}{=} \max_{u \in V} \sum_{v \in V} \rho_{\mu}(u \rightarrow v) < 1.
\end{align*}
\end{definition}
A straightforward application of the path coupling technique of \cite{BD97i, BD97ii} shows that if $\norm{\rho_{\mu}}_{1} < 1$, then there is a coupling for the Glauber dynamics which is one-step contractive w.r.t. Hamming distance. We state this well-known implication formally here, and refer to \cite{DGJ09} for the proof.
\begin{fact}\label{fact:dobrushinpathcoupling}
Let $\mu$ be a distribution on some finite product space $\prod_{v \in V} \Omega(v)$, where $V$ is a finite index set. If $\norm{\rho_{\mu}}_{1} \leq \gamma < 1$, then the Glauber dynamics is $(1-\alpha)$-contractive w.r.t. Hamming distance with $\alpha = \frac{1}{n}(1 - \gamma)$.
\end{fact}
In particular, combining \cref{thm:curvaturespecind} and \cref{fact:dobrushinpathcoupling} immediately yields spectral independence under the Dobrushin uniqueness condition. Combined with \cref{thm:localtoglobalspinoptimal}, this additionally recovers a version of a result due to \cite{Mar19}, which says that a weaker $\ell_{2}$-version of the Dobrushin uniqueness condition (see also \cite{Hay06, DGJ09}) implies a $\Omega(1/n)$ log-Sobolev constant for the Glauber dynamics.

\begin{corollary}[Dobrushin Uniqueness Implies Spectral Independence]
Let $\mu$ be a distribution on some finite product space $\prod_{v \in V} \Omega(v)$, where $V$ is a finite index set. If $\norm{\rho_{\mu}}_{1} \leq \gamma < 1$, then $\mu$ is $(\eta_{0},\dots,\eta_{n-2})$-spectrally independent with $\eta_{k} \leq \frac{4}{1 - \gamma} - 1$ for all $k$. If additionally the measure $\mu$ is the Gibbs distribution of a spin system on a bounded-degree graph, then the spectral gap, standard and modified log-Sobolev constants for the Glauber dynamics are all $\Omega(1/n)$.
\end{corollary}

\section{Spectral Independence for Proper List-Colorings}\label{sec:listcolorings}
We now specialize to the setting of proper list-colorings of a graph. Formally, we fix a graph $G=(V,E)$, a collection of color lists $(L(v))_{v \in V}$. We call a configuration $\sigma \in \prod_{v \in V} L(v)$ a list-coloring of $G$. We say a list-coloring $\sigma$ is proper if $\sigma(u) \neq \sigma(v)$ whenever $u \neq v$ are neighbors. Throughout, we will let $\Delta$ denote the maximum degree of $G$, and we assume $\Delta \leq O(1)$. We also assume there is a positive integer $q \geq \Delta + 2$ such that $L(v) \subseteq [q]$ for all $v \in V$.

A well-known result due to \cite{Jer95} using path coupling shows that if $|L(v)| > 2\Delta$ for all $v \in V$, then there is a contractive one-step coupling for the Glauber dynamics which yields $O(n\log n)$ mixing. As noted in \cite{CLV21}, one can adapt the argument of \cite{GKM15} to obtain strong spatial mixing when $|L(v)| > 2\Delta$, and use the arguments of \cite{CGSV21, FGYZ21} to deduce spectral independence in this regime. However, it is still open whether one can obtain strong spatial mixing below the $2\Delta$ threshold; see \cite{GKM15, EGHSV19} for results going below $2\Delta$ on special classes of graphs.

In the seminal work of Vigoda \cite{Vig00}, it was shown that there is a contractive one-step coupling for a different local Markov chain known as the flip dynamics whenever $|L(v)| \geq \frac{11}{6}\Delta$. This threshold was further improved to $|L(v)| \geq \wrapp{\frac{11}{6} - \epsilon}\Delta$ in a recent breakthrough by \cite{CDMPP19}, this time using a more sophisticated variable-length coupling. Both works further showed that Glauber dynamics mixes in $O(n^{2})$ time in this regime using a spectral gap comparison argument \cite{DS93}.

Our goal is to use these coupling results along with \cref{thm:specindlocalamortized} to obtain spectral independence for the uniform distribution over proper list-colorings in the regime $|L(v)| \geq \wrapp{\frac{11}{6} - \epsilon}\Delta$. Combined with \cref{thm:localtoglobalspinoptimal}, we improve the previous $O(n^{2})$ mixing time bound to the optimal $O(n\log n)$, as well as show Chernoff-type concentration bounds for Lipschitz functions, which were not known before.
\subsection{The Flip Dynamics}
We follow the presentation in \cite{CDMPP19}, which generalizes the flip dynamics analyzed in \cite{Vig00} to list-colorings. Fix a list-coloring $\sigma$. We say a path $u = w_{1},\dots,w_{\ell} = v$ in $G$ is an alternating path from $u$ to $v$ using colors $\sigma(u), c$ if for all $i$, we have $\sigma_{w_{i}} \in \{\sigma(u),c\}$ and $\sigma_{w_{i}} \neq \sigma_{w_{i+1}}$. For a fixed list-coloring $\sigma$, $v \in V$ and color $c$, we define the \textbf{Kempe component for $\sigma,v,c$} by the following subset of vertices.
\begin{align*}
    S_{\sigma}(u,c) = \wrapc{v \in V : \substack{\exists \text{ alternating path from }u \\ \text{to } v \text{ using } \sigma(u),c}}.
\end{align*}
Given $\sigma$ and a Kempe component $S = S_{\sigma}(u,c)$, we define $\sigma_{S}$ to be the coloring obtained by ``flipping'' the color assigned to vertices in $\{v \in S : \sigma(v) = \sigma(u)\}$ to $c$, and the color assigned to vertices in $\{v \in S : \sigma(v) = c\}$ to $\sigma(u)$. Note that $\sigma_{S}$ need not be a proper list-coloring; we say a Kempe component $S = S_{\sigma}(u,c)$ is \textbf{flippable} if the coloring $\sigma_{S}$ is a proper list-coloring.

For each $j \in \N$, let $0 \leq p_{j} \leq 1$ be a tunable parameter to be determined later. We define the flip dynamics with flip parameters $\{p_{j}\}_{j \in \N}$ for sampling proper list-colorings as follows: Given the current list-coloring $\sigma^{(t-1)}$, we generate the next list-coloring $\sigma^{(t)}$ by the following two-step process:
\begin{enumerate}
    \item Select a uniformly random vertex $v^{(t)} \in V$, and a uniformly random color $c^{(t)} \in L(v^{(t)})$.
    \item If the Kempe component $S = S_{\sigma^{(t-1)}}(v^{(t)}, c^{(t)})$ is flippable, set $\sigma^{(t)} = \sigma_{S}^{(t-1)}$ with probability $\frac{p_{j}}{j}$ and $\sigma^{(t)} = \sigma^{(t-1)}$ otherwise, where $j = |S|$.
\end{enumerate}
We write $P_{\mu, \flip}$ for the transition probability matrix of the flip dynamics. It is straightforward to verify that the stationary distribution of the flip dynamics is uniform over proper list-colorings, regardless of the choice of the flip parameters. One can recover the Wang-Swendsen-Koteck\'{y} Markov chain by setting $p_{j} = j$ for all $j \in \N$ \cite{WSK89}.

\cite{Vig00} showed that with flip parameters
\begin{align}\label{eq:Vigodaflip}
    p_{1} = 1 \quad p_{2} = \frac{13}{42} \quad p_{3} = \frac{1}{6} \quad p_{4} = \frac{2}{21} \quad p_{5} = \frac{1}{21} \quad p_{6} = \frac{1}{84} \quad p_{j} = 0, \forall j \geq 7,
\end{align}
there is a one-step coupling which is contractive w.r.t. Hamming distance whenever $\abs{L(v)} \geq \frac{11}{6}\Delta$. \cite{CDMPP19} showed using linear programming arguments that this is optimal in the sense that when $\abs{L(v)} < \frac{11}{6}$, there is no choice of the flip parameters which has a one-step contractive coupling w.r.t. Hamming distance. They additionally construct an explicit family of hard instances witnessing optimality.

One of the key insights of \cite{CDMPP19} is that the optimal choice of flip parameters comes out of the solution to a linear program, with the objective value of the program governing the contraction properties of the coupling. By solving this linear program, they show that for the following choice of flip parameters
\begin{align}\label{eq:CDMPPflip}
    \hat{p}_{1} = 1 \quad \hat{p}_{2} \approx 0.296706 \quad \hat{p}_{3} \approx 0.166762 \quad \hat{p}_{4} \approx 0.101790 \quad \hat{p}_{5} \approx 0.058475 \quad \hat{p}_{6} = 0.025989 \quad p_{j} = 0, \forall j \geq 7,
\end{align}
there is a variable-length coupling such that the Hamming distance contracts by a constant factor every $O(n)$ steps in expectation. One can thus expect that the coupling is $C$-amortized convergent with $C \leq O(n)$.

We formalize their main coupling result in the following subsection. For the moment, we state two intermediate lemmas, prove one of them, and show how they imply \cref{thm:coloringspecindmain}.
\begin{lemma}\label{lem:flipmatrixdifbound}
Assume the input graph $G=(V,E)$ has maximum degree $\Delta \leq O(1)$. Then, the flip dynamics with parameters given in \cref{eq:CDMPPflip} satisfy the following:
\begin{align*}
    \max_{\tau \in \supp(\mu \mid uc)} \wrapc{\sum_{\sigma \neq \tau} \abs{P_{\mu,\flip}(\tau \rightarrow \sigma) - P_{\mu \mid uc, \flip}(\tau \rightarrow \sigma)}} \leq O(1/n).
\end{align*}
\end{lemma}
\begin{lemma}\label{lem:flipdistbound}
Let $(G,\mathcal{L})$ be a list-coloring instance, where $\Delta \leq O(1)$ and $|L(v)| \geq \lambda^{*}\Delta$ for all $v \in V$, where $\lambda^{*} = \frac{11}{6} - \epsilon$ and $\epsilon \approx 10^{-5}$ is a small constant. Then the flip dynamics with parameters given in \cref{eq:CDMPPflip} admits a $C$-amortized convergent coupling w.r.t. Hamming distance where $C \leq O(n)$.
\end{lemma}
\begin{proof}[Proof of \cref{thm:coloringspecindmain}]
The flip dynamics is clearly $O(1)$-local w.r.t. Hamming distance since only Kempe components of size at most $6$ can be flipped. $(\eta_{0},\dots,\eta_{n-2})$-spectral independence where $\eta_{k} \leq O(1)$ for all $k$ then follows immediately by combining \cref{lem:flipmatrixdifbound} and \cref{lem:flipdistbound} with \cref{thm:blackboxcomparison}. The lower bounds on the spectral gap, standard and modified log-Sobolev constants then follow from \cref{thm:localtoglobalspinoptimal}.
\end{proof}
\begin{proof}[Proof of \cref{lem:flipmatrixdifbound}]
The main detail one must be careful of is that the flip dynamics for sampling from $\mu \mid uc$ always leaves the color for $u$ fixed to $c$. Hence, flipping any Kempe component containing $u$ leads to potentially different list-colorings under $P_{\mu, \flip}$ versus $P_{\mu \mid uc, \flip}$. However, since we only flip components of $O(1)$-size, this isn't an issue for us.

Fix a $\tau$ with $\tau(u) = c$, and let $B(u,6)$ denote the set of vertices of shortest path distance at most $6$ away from $u$ in $G$. Since we only flip Kempe components of size at most $6$, we have that for any $v \in V \setminus B(u,6)$ and $c \in L(u)$, the flippable Kempe component $S_{\tau}(v,c')$ does not contain $u$, and hence, flipping it leads to the same list-coloring under $P_{\mu,\flip}$ and $P_{\mu \mid vc, \flip}$. Hence, we have
\begin{align*}
    &\sum_{\sigma \neq \tau} \abs{P_{\mu,\flip}(\tau \rightarrow \sigma) - P_{\mu \mid uc, \flip}(\tau \rightarrow \sigma)} \\
    &= \sum_{v \in V : v \notin B(u,6)} \sum_{c' \in L(v)} \frac{1}{|L(v)|} \cdot \wrapp{\frac{1}{n} - \frac{1}{n-1}} \cdot p_{|S_{\tau}(v,c')|} + \sum_{v \in B(u,6)} \sum_{c' \in [q]} \abs{P_{\mu,\flip}(\tau \rightarrow \sigma) - P_{\mu \mid uc, \flip}(\tau \rightarrow \sigma)} \\
    &\leq \frac{n - |B(u,6)|}{n(n-1)} + \frac{|B(u,6)|}{n} \\
    &\leq \frac{|B(u,6)| + 1}{n} \\
    &\lesssim \frac{\Delta^{6}}{n} \\
    &\leq O(1/n). \tag{Bounded-degree assumption}
\end{align*}
\end{proof}
\begin{remark}
As one can see in the proof from the factor of $\Delta^{6}$, we have made no attempt to optimize constants.
\end{remark}
At this point, all that remains is to prove \cref{lem:flipdistbound}, which we do using the variable-length path coupling constructed in \cite{CDMPP19}.

\subsection{Variable-Length Path Coupling: Proof of \texorpdfstring{\cref{lem:flipdistbound}}{flipdistbound}}
To begin, we first define the notion of variable-length coupling following \cite{HV07, CDMPP19}.
\begin{definition}[Path-Generating Set]
For a finite state space $\Omega$, a \textbf{path generating set} is a subset $S \subseteq \binom{\Omega}{2}$ such that the undirected graph $(\Omega,S)$ is connected. We let $d_{S}(\cdot,\cdot)$ denote the induced shortest-path metric on $\Omega$, and write $d(\cdot,\cdot)$ when the path generating set $S$ is clear from context. We also write $x \sim y$ whenever $\{x,y\} \in S$.
\end{definition}
\begin{definition}[Variable-Length Path Coupling \cite{HV07}]
Fix an irreducible transition probability matrix $P$ which is reversible w.r.t. a distribution $\pi$ on a finite state space $\Omega$, and let $d(\cdot,\cdot)$ be a metric on $\Omega$ induced by a path generating set $S \subseteq \binom{\Omega}{2}$. For every pair of starting states $x^{(0)},y^{(0)} \in \Omega$ with $x^{(0)} \sim y^{(0)}$, we let $(\overline{x}, \overline{y}, T) = (\overline{x}(x^{(0)}, y^{(0)}), \overline{y}(x^{(0)},y^{(0)}), T(x^{(0)},y^{(0)}))$ denote a random variable where $T$ is a (potentially random) nonnegative integer and $\overline{x} = (x^{(0)}, x^{(1)},\dots,x^{(T)}), \overline{y} = (y^{(0)}, y^{(1)},\dots,y^{(T)})$ are length-$T$ sequences of states in $\Omega$.

For every integer $t \geq 0$ and every pair of neighboring states $x^{(0)} \sim y^{(0)}$, define random variables $x_{t},y_{t}$ by the following experiment. Sample $(\overline{x},\overline{y},T)$, and set $x_{t} = x^{(t)}, y_{t} = y^{(t)}$ if $t \leq T$, and sample $x_{t} \sim P^{t - T}(x^{(T)},\cdot), y_{t} \sim P^{t - T}(y^{(T)},\cdot)$ if $t > T$. We say the random variable $(\overline{x},\overline{y},T)$ is a \textbf{variable-length path coupling for $P$} if $x_{t} \sim P^{t}(x^{(0)},\cdot), y_{t} \sim P^{t}(y^{(0)},\cdot)$ for every integer $t \geq 0$ and every pair of neighboring states $x^{(0)} \sim y^{(0)}$. In this case, we say that $\overline{x},\overline{y}$ are individually \textbf{faithful copies}. If $T = t$ with probability $1$ for some nonnegative integer $t \geq 0$, we say that $(\overline{x},\overline{y},T)$ is a \textbf{$t$-step path coupling}.
\end{definition}
\begin{remark}
In our application to colorings, the random time $T$ will be a stopping time in the sense that its value only depends on the past, i.e. $x^{(0)},y^{(0)},\dots, x^{(t)},y^{(t)}$ for $t \leq T$.
\end{remark}
Given a variable-length path coupling, \cite{HV07} showed one can construct a full coupling, generalizing the original path coupling theorem of \cite{BD97i, BD97ii}. Furthermore, the contraction properties of the full coupling are inherited from the path coupling. While the original statement in \cite{HV07} merely states rapid mixing given a variable-length path coupling, its proof implies the following.
\begin{theorem}[Proof of Corollary 4 from \cite{HV07}]\label{thm:variabletomultistep}
Let $(\overline{x},\overline{y},T)$ be a variable-length path coupling w.r.t. a path generating set $S$ for a reversible Markov chain $P$ on a state space $\Omega$ with stationary distribution $\pi$. Let
\begin{align*}
    \alpha \overset{\defin}{=} 1 - \max_{\{x^{(0)}, y^{(0)}\} \in S} \E[d_{H}(x^{(T)}, y^{(T)})] \quad\quad W \overset{\defin}{=} \max_{\{x^{(0)}, y^{(0)}\} \in S, t \leq T} d_{H}(x^{(t)}, y^{(t)}) \quad\quad \beta \overset{\defin}{=} \max_{\{x^{(0)}, y^{(0)}\} \in S} \E[T].
\end{align*}
Assume $0 < \alpha < 1$. Then there is a full $M$-step coupling with $M = \lceil \frac{2\beta W}{\alpha} \rceil$ such that for all pairs $x^{(0)},y^{(0)}$, which need not be neighbors in $S$, we have the inequality
\begin{align*}
    \E[d_{H}(x^{(M)},y^{(M)}) \mid x^{(0)},y^{(0)}] \leq \wrapp{1 - \frac{\alpha}{2}} \cdot d_{H}(x^{(0)},y^{(0)}).
\end{align*}
\end{theorem}
Given this, all we need now is a good variable-length path coupling. This is given by the following result due to \cite{CDMPP19}.
\begin{theorem}[\cite{CDMPP19}]\label{thm:coloringvarlengthcoupling}
Let $(G,\mathcal{L})$ be a list-coloring instance, where $G=(V,E)$ is a graph with maximum degree $\Delta \leq O(1)$, and $\mathcal{L} = (L(v))_{v \in V}$ is a collection of color lists. Let the path generating set $S$ be given by the set of pairs $\{\tau,\sigma\}$ such that $\tau,\sigma$ differ on the coloring of exactly one vertex. Assume $|L(v)| \geq \lambda^{*}\Delta$ for all $v \in V$ where $\lambda^{*} = \frac{11}{6} - \epsilon$ for an absolute constant $\epsilon \approx 10^{-5}$. Then there exists a variable-length path coupling $(\overline{\tau},\overline{\sigma},T)$ for the flip dynamics w.r.t. $S$ with flip parameters given in \cref{eq:CDMPPflip}, where $T$ is the first time such that the Hamming distance changes, such that $\alpha = \frac{q - \lambda^{*}\Delta}{q - \Delta - 2} = \Theta(1)$, $W = 13$ and $\beta \leq \frac{qn}{q - \Delta - 2} \leq O(n)$
\end{theorem}
With these tools in hand, we may now finally prove \cref{lem:flipdistbound} and complete the proof of \cref{thm:coloringspecindmain}.
\begin{proof}[Proof of \cref{lem:flipdistbound}]
First, note that the path generating set $S$ generates the Hamming metric $d_{H}(\cdot,\cdot)$ on proper list-colorings. Now, given the variable-length path coupling furnished by \cref{thm:coloringvarlengthcoupling}, we use \cref{thm:variabletomultistep} to construct an $M$-step coupling with $M = \lceil \frac{2\beta W}{\alpha} \rceil \leq O(n)$ which contracts with rate $1 - \alpha$ every $M$ steps, where $\alpha$ is is a constant independent of $n$. Under this coupling, for every $k = 0,\dots,M-1$ and every positive integer $j$, we have that
\begin{align*}
    &\E_{\tau^{(jM + k)}, \sigma^{(jM + k)}}\wrapb{d_{H}(\tau^{(jM+k)}, \sigma^{(jM+k)}) \mid \tau^{(k)},\sigma^{(k)}} \\
    &\leq \wrapp{1 - \frac{\alpha}{2}}\E_{\tau^{((j-1)M + k)}, \sigma^{((j-1)M + k)}}\wrapb{d_{H}(\tau^{((j-1)M+k)}, \sigma^{((j-1)M+k)}) \mid \tau^{(k)},\sigma^{(k)}} \\
    &\leq \dots \\
    &\leq \wrapp{1 - \frac{\alpha}{2}}^{j} \cdot d_{H}(\tau^{(k)}, \sigma^{(k)}),
\end{align*}
where $\tau^{(0)} = \tau, \sigma^{(0)} = \sigma$ are arbitrary starting states, which need not be neighbors under $S$. It follows that
\begin{align*}
    \sum_{t=0}^{\infty} \E_{\tau^{(t)},\sigma^{(t)}}\wrapb{d_{H}(\tau^{(t)}, \sigma^{(t)}) \mid \substack{\tau^{(0)} = \tau \\ \sigma^{(0)} = \sigma}} &\leq \sum_{k=0}^{M-1} \sum_{j=0}^{\infty} \E_{\tau^{(jM + k)}, \sigma^{(jM + k)}}\wrapb{d_{H}(\tau^{(jM+k)}, \sigma^{(jM+k)}) \mid \tau^{(k)},\sigma^{(k)}} \\
    &\leq \sum_{k=0}^{M-1} \E\wrapb{d_{H}(\tau^{(k)},\sigma^{(k)}) \mid \tau^{(0)}, \sigma^{(0)}} \sum_{j=0}^{\infty} \wrapp{1 - \frac{\alpha}{2}}^{j} \\
    &= \frac{2}{\alpha} \sum_{k=0}^{M-1} \E\wrapb{d_{H}(\tau^{(k)},\sigma^{(k)}) \mid \tau^{(0)}, \sigma^{(0)}} \\
    &\leq \frac{2M}{\alpha} d_{H}(\tau^{(0)},\sigma^{(0)}) \tag*{$(\ast)$} \\
    &\leq O(n) \cdot d_{H}(\tau^{(0)}, \sigma^{(0)}).
\end{align*}
To justify $(\ast)$, note that $T$ is the first time the Hamming distance changes, and that each time the Hamming distance changes, the expected Hamming distance contracts by a factor of $1 - \alpha$.
\end{proof}

\section{Future Directions}
Two concrete open problems are to bring down the required number of colors from $\wrapp{\frac{11}{6} - \epsilon}\Delta$ to $\Delta + 2$, and to remove the bounded-degree assumption, both in this work and in \cite{CLV21}. Another interesting question is if spectral independence implies any useful notion of correlation decay, such as strong spatial mixing, or the absence of zeros for partition function in a large region. This is relevant particularly for proper list-colorings, where we showed spectral independence when $q \geq \wrapp{\frac{11}{6} - \epsilon}\Delta$, but correlation decay and absence of zeros are both open in general when number of colors is below $2\Delta$.

We also reiterate that one feature of our approach which we haven't exploited is that in order to obtain $O(1)$-spectral independence, it suffices for the Markov chain admitting the nice coupling to merely update $O(1)$-coordinates in a single move ``on average'', as opposed to the worst-case starting state; see \cref{rem:localflexible}. We leave it to future work to see if this can be exploited.

\pagebreak
\printbibliography
\pagebreak
\begin{appendices}
\section{Variance and Entropy Decay}\label{app:varent}
While we primarily use prior results on the spectral gap, and standard and modified log-Sobolev constants as blackboxes, to keep this paper self-contained, we define these constants here, and state their relevance to mixing and concentration. Fix a Markov kernel $P$ on a finite state space $\Omega$ which reversible w.r.t. a distribution $\pi$. We may define an inner product using $\pi$ by $\langle f,g \rangle_{\pi} = \E_{\pi}[fg]$. This inner product together with the kernel $P$ induces a positive semidefinite quadratic form known as the Dirichlet form, defined as $\mathcal{E}_{P}(f,g) = \langle f, (I - P)g \rangle_{\pi}$. The variance of a function $f:\Omega \rightarrow \R$ is given by $\var_{\pi}(f) = \E_{\pi}(f^{2}) - \E_{\pi}(f)^{2}$, while the entropy of a function $f:\Omega \rightarrow \R_{\geq0}$ is given by $\Ent_{\pi}(f) = \E_{\pi}(f \log f) - \E_{\pi}(f) \log \E_{\pi}(f)$.

With these notions in hand, we may now define the following constants:
\begin{flalign}\label{eq:mixingconstants}
\begin{split}
    \text{(Spectral Gap)} \quad\quad \lambda(P) &\overset{\defin}{=} \inf_{f \neq 0} \frac{\mathcal{E}(f,f)}{\var_{\pi}(f)} \\
    \text{(Modified Log-Sobolev Constant)} \quad\quad \rho(P) &\overset{\defin}{=} \inf_{f \geq 0} \frac{\mathcal{E}(f, \log f)}{\Ent_{\pi}(f)} \\
    \text{(Standard Log-Sobolev Constant)} \quad\quad \kappa(P) &\overset{\defin}{=} \inf_{f \geq 0} \frac{\mathcal{E}(\sqrt{f},\sqrt{f})}{\Ent_{\pi}(f)}.
\end{split}
\end{flalign}
It is known that $4\kappa(P) \leq \rho(P) \leq 2\lambda(P)$ \cite{BT03}, with lower bounds on $\kappa(P)$ being the most difficult to establish. For the reader's convenience, we collect some well-known relations between these constants, mixing, and concentration.
\begin{proposition}[Mixing and Concentration]
We have the following bounds on the mixing time of a Markov chain with transition probability matrix $P$ and stationary distribution $\pi$.
\begin{align*}
    t_{\mix}(\epsilon) &\leq \frac{1}{\lambda(P)} \wrapp{\frac{1}{2}\log \frac{1}{\pi_{\min}} + \log \frac{1}{2\epsilon}} \tag*{\cite{LPW17}} \\
    t_{\mix}(\epsilon) &\leq \frac{1}{\rho(P)} \wrapp{\log\log \frac{1}{\pi_{\min}} +\log \frac{1}{2\epsilon^{2}}} \tag*{\cite{BT03}} \\
    t_{\mix}(\epsilon) &\leq \frac{1}{4\kappa(P)} \wrapp{\log\log \frac{1}{\pi_{\min}} + \log \frac{1}{2\epsilon^{2}}} \tag*{\cite{DS96}}.
\end{align*}
Furthermore, for every function $f:\Omega \rightarrow \R$ which is $1$-Lipschitz w.r.t. graph distance under $P$, we have the following Chernoff-type concentration inequalities \cite{Goe04, Sam05, BLM16}.
\begin{align*}
    \Pr[f \geq \E_{\pi}(f) + \epsilon] &\leq \exp\wrapp{-\frac{\rho(P) \epsilon^{2}}{2 v(f)}} \\
    \Pr[f \geq \E_{\pi}(f) + \epsilon] &\leq \exp\wrapp{-\frac{\kappa(P) \epsilon^{2}}{2v(f)}},
\end{align*}
where
\begin{align*}
    v(f) \overset{\defin}{=} \max_{x \in \Omega}\wrapc{\sum_{y \in \Omega} P(x \rightarrow y) \cdot (f(x) - f(y))^{2}}.
\end{align*}
\end{proposition}
\end{appendices}
\end{document}